\documentclass[A4paper,notitlepage,superscriptaddress,longbibliography,aps,reprint,10pt]{revtex4-1}
\usepackage{graphicx}
\usepackage{amsfonts}
\usepackage{amssymb}
\usepackage{mathrsfs}
\usepackage{mathtools}
\usepackage{soul}
\usepackage[usenames, dvipsnames]{color}
\usepackage[colorlinks=true,citecolor=blue]{hyperref}

\usepackage{tikz}

\usepackage[mathcal]{euscript}

\newcommand{\mbb}[1]{\mathbb{#1}}
\newcommand{\mcl}[1]{\mathcal{#1}}

\usepackage{amsmath}

\renewcommand{\d}{\mathrm{d}}
\newcommand{\Tr}{\operatorname{Tr}}

\newcommand{\ket}[1]{\left| #1 \right\rangle}

\newcommand{\ketbra}[2]{\left| #1 \right\rangle \left\langle #2 \right|}

\newcommand{\dprod}[2]{\left\langle #1, #2\right\rangle}

\newcommand{\abs}[1]{\left| #1\right|}

\usepackage{amsthm}
\newtheorem{lemma}{Lemma}
\newtheorem*{lemma*}{Lemma}

\newtheorem{theorem}{Theorem}
\newtheorem{conjecture}{Conjecture}
\newtheorem{corollary}{Corollary}
\newtheorem*{corollary*}{Corollary}
\theoremstyle{remark}

\usepackage{graphicx}

\usepackage{float}
\usepackage{color}
\definecolor{npurple}{rgb}{0.3,0,0.6}

\renewcommand{\S}{\mathcal{S}}
\newcommand{\N}{\mcl{N}}
\newcommand{\M}{\mcl{M}}

\newcommand{\II}{\mbb{I}}

\newcommand{\CC}{\mbb{C}}
\newcommand{\C}{\mcl{C}}
\renewcommand{\P}{\mcl{P}}

\newcommand{\K}{\mcl{K}}

\renewcommand{\H}{\mcl{H}}
\newcommand{\A}{\mcl{A}}
\newcommand{\B}{\mcl{B}}
\newcommand{\X}{\mcl{X}}
\newcommand{\Y}{\mcl{Y}}
\renewcommand{\L}{\mcl{L}}
\newcommand{\n}{\pmb{n}}

\newcommand{\G}{\mathcal{G}}

\definecolor{mygray}{gray}{0.6}

\begin{document}
\title{Quantum steering with positive operator valued measures}
\date{\today}
\author{H. Chau Nguyen}
\email{chau@pks.mpg.de}
\affiliation{Max-Planck-Institut f\"ur Physik komplexer Systeme,  N\"othnitzer Stra{\ss}e 38, D-01187 Dresden, Germany}
\affiliation{Naturwissenschaftlich-Technische Fakult\"at, Universit\"at Siegen, Walter-Flex-Stra{\ss}e 3, D-57068 Siegen, Germany}

\author{Antony Milne}
\affiliation{Department of Physics, Imperial College London, London SW7 2AZ, United Kingdom}
\affiliation{Department of Computing, Goldsmiths, University of London, New Cross, London SE14 6NW, United Kingdom}

\author{Thanh Vu}
\affiliation{Hanoi University of Science and Technology, 1 Dai Co Viet, Hai Ba Trung, Hanoi, Vietnam}

\author{Sania Jevtic}
\affiliation{Department of Mathematics, Imperial College London, London SW7 2AZ, United Kingdom}

\begin{abstract}
We address the problem of quantum nonlocality with positive operator valued measures (POVM) in the context of Einstein-Podolsky-Rosen quantum steering. We show that, given a candidate for local hidden state (LHS) ensemble, the problem of determining the steerability of a bipartite quantum state of finite dimension with POVMs can be formulated as a nesting problem of two convex objects. One consequence of this is the strengthening of the theorem that justifies choosing the LHS ensemble based on symmetry of the bipartite state. As a more practical application, we study the classic problem of the steerability of two-qubit Werner states with POVMs. We show strong numerical evidence that these states are unsteerable with POVMs up to a mixing probability of $\frac{1}{2}$ within an accuracy of $10^{-3}$. 
\end{abstract}


\maketitle

\textit{Introduction.}
Ever since its first examination by Einstein, Podolsky, and Rosen (EPR) in 1935~\cite{epr1935}, quantum nonlocality has been a puzzling phenomenon. In the EPR thought experiment, one observer, Alice, can perform a measurement on her half of an entangled pair to steer the other half (that belongs to a distant observer, Bob) to ensembles that conflict with the very intuition of classical locality~\cite{schrodinger1936}. This conflict was so profound that it prompted EPR to conclude that quantum theory was ``incomplete''~\cite{epr1935} and caused the longest debate in the history of quantum mechanics~\cite{bell2004speakable}. A more ``complete'' theory would be supplemented by hidden variables; however, the seminal work of Bell~\cite{bell1964} demonstrated that no such theory, when constrained by locality, is capable of explaining all quantum mechanical predictions for bipartite systems. Nowadays, quantum nonlocality is perceived as one of the hallmarks of quantum theory that sets it apart from classical notions and underlies numerous quantum information applications~\cite{brunner2014a}.

Bell's work defined the first class of quantum nonlocality, now known as Bell nonlocality~\cite{brunner2014a}. Some 25 years later, Werner realised that Bell nonlocality and entanglement (nonseparability) were in fact two independent forms of quantum nonlocality~\cite{werner1989quantum}. In 2007, Wiseman, Jones and Doherty~\cite{wiseman2007steering} recognised that the original idea of the EPR thought experiment is actually best captured by yet another form of quantum nonlocality -- quantum steerability. Since then, quantum steerability has been successfully demonstrated experimentally in loophole-free tests~\cite{bennet2012,smith2012,wittmann2012}. It has been employed in a range of practical quantum information tasks, including quantum cryptography~\cite{branciard2012}, randomness certification~\cite{law2014,passaro2015}, and self-testing~\cite{supic2016,kashefi2017}. 
 
Among this surge of discoveries, a fundamental question remains: which bipartite states manifest quantum steerability? In fact, determining the steerability of a bipartite state when considering all possible measurements, i.e., positive operator valued measure (POVM) measurements, has been such a challenging task that it is unanswered for even the simplest case of two-qubit Werner states~\cite[Problem 39]{oqp}. The problem remains open in spite of many significant advances towards understanding quantum steering under particular subsets of POVMs, e.g., with projection valued measure (PVM) measurements~\cite{wiseman2007steering,jevtic2015einstein,nguyen2016necessary}, with finite subsets of POVMs~\cite{cavalcanti2016lhsmalg,hirsch2016lhsmalg}, and with highly noisy POVMs or highly noisy states~\cite{quintino2015wernerlhsm}.

{In this Letter, we are concerned with the problem of quantum steering with POVMs for bipartite systems of arbitrary (but finite) dimension. We demonstrate that for a given choice of local hidden state ensemble, the task of determining whether a quantum state is steerable can be considered as a nesting problem of two convex objects. As a consequence, we derive an inequality which allows a test of steerability for all measurements. Surprisingly, the inequality also reveals a fundamental aspect of quantum steering. Namely, in quantum steering, the choice of local hidden variable is no longer arbitrary as in Bell nonlocality, but can be limited to the set of Bob's pure states. This in fact makes the study of quantum steering significantly simpler than its partner Bell nonlocality. In particular, one can strengthen the theorem (Lemma 1 of Ref.~\cite{wiseman2007steering}) which limits the choice of local hidden state ensemble based on the symmetry of the state. As the first application, we then apply the inequality to study the steerability of the two-qubit Werner states. 
Contrary to the fact that general POVMs provide an advantage over PVMs in many situations~\cite{vertesi2010bellpovm,hirsch2013,acin2016}, we provide strong numerical evidence that POVMs and PVMs are in fact equivalent for steering two-qubit Werner states.}

\textit{Quantum steerability.}
Suppose Alice and Bob share a bipartite quantum state $\rho$ over the finite-dimensional Hilbert space $\H_A \otimes \H_B$. We use $\A^H$ (or $\B^H$) to denote the space of Hermitian operators  over $\H_A$ (or $\H_B$). A POVM measurement with $n$ outcomes ($n$-POVM) $E$ implemented by Alice is an (ordered) collection of $n$ positive operators, $E= \{E_i\}_{i=1}^{n}$ with $E_i \in \A^H$, $E_i \ge 0$ and $\sum_{i=1}^{n} E_i = \II_A$, where $\II_A$ is the identity operator on $\A^H$. On performing the measurement, Alice steers Bob's system to the \emph{steering ensemble}  $\{\Tr_A [\rho (E_i \otimes \II_B)]\}_{i=1}^{n}$. However, despite the arbitrary choice of measurements, for certain bipartite states, the steering experiment can be locally simulated. More specifically, let $u$ be an ensemble (that is, a probability distribution) on the set of Bob's pure states, denoted by $\mathcal{S}_B$. A state $\rho$ is then called $u$-\emph{unsteerable} (always considered from Alice's side) with respect to $n$-POVMs if, for any $n$-POVM $E$, Alice can find $n$ response functions $G_i$ (with $G_i(P) \ge 0$ and $\sum_{i=1}^{n} G_i(P)= 1$ for all $P \in \mathcal{S}_B$) such that the steering ensemble can be simulated via a \emph{local hidden state model}~\cite{wiseman2007steering},
\begin{equation}
\Tr_{A}[\rho (E_i \otimes \II_B)]= \int \d \omega (P) u(P) \,\,  G_i(P) P,
\label{eq:unsteerable_def}
\end{equation}
where the integral is taken over the Haar measure $\omega$ on Bob's pure states $\S_B$. 
Equation~\eqref{eq:unsteerable_def} ensures that Bob, when performing state tomography conditioned on Alice's outcomes, obtains the same result as if Alice were steering his system~\cite{wiseman2007steering}.
The ensemble $u$ is called a \emph{local hidden state} (LHS) ensemble. In principle, the domain of the ensemble $u$ can be extended to mixed states. However, as a mixed state can be written as a convex combination of pure states, restricting the domain of $u$ to pure states causes no loss of generality.
We say then that $\rho$ is unsteerable with $n$-POVMs if there exists $u$ such that $\rho$ is $u$-unsteerable with $n$-POVMs. 

We note that this definition of quantum steering is slightly different from the original definition~\cite{wiseman2007steering}. In the latter, the LHS ensemble is indexed by a local hidden variable. We will prove the two definitions are equivalent as parts of our results. Our seemingly minor simplification in fact has very important consequences, which will be discussed below.


{\it The set of $n$-POVMs and its geometry.} 
The key idea in our approach is that an $n$-POVM $E$ can be thought of as a point in the real vector space of composite operators $(\A^H)^{\oplus n} = \oplus_{i=1}^{n} \A^H$. We therefore write $E=\oplus_{i=1}^{n} E_i$, which explicitly indicates that it is a composite operator in $(\A^H)^{\oplus n}$ with components $E_i$ each bounded by $0 \leq E_i \leq \II_A$. The space $(\A^H)^{\oplus n}$ can be made Euclidean by defining an inner product $\dprod{X}{Y} = \sum_{i=1}^{n} \dprod{X_i} {Y_i}$ for any composite operators $X$ and $Y$, where $\dprod{X_i}{Y_i}$ denotes the Hilbert--Schmidt inner product of $\A^H$, $\dprod{X_i}{Y_i} = \Tr (X^{\dagger}_iY_i)$. The set of $n$-POVMs is then a convex and compact subset of this space~\cite{Chiribella2007a}, which we denote by $\M^n$. Since $\sum_{i=1}^{n} E_i = \II_A$, $\M^n$ in fact belongs to the linear manifold $\P^n=\{X \vert \sum_{i=1}^n X_i=\II_A\}$.

\begin{figure}[!ht]
\begin{minipage}[c]{0.15\textwidth}
\includegraphics[width=0.9\textwidth]{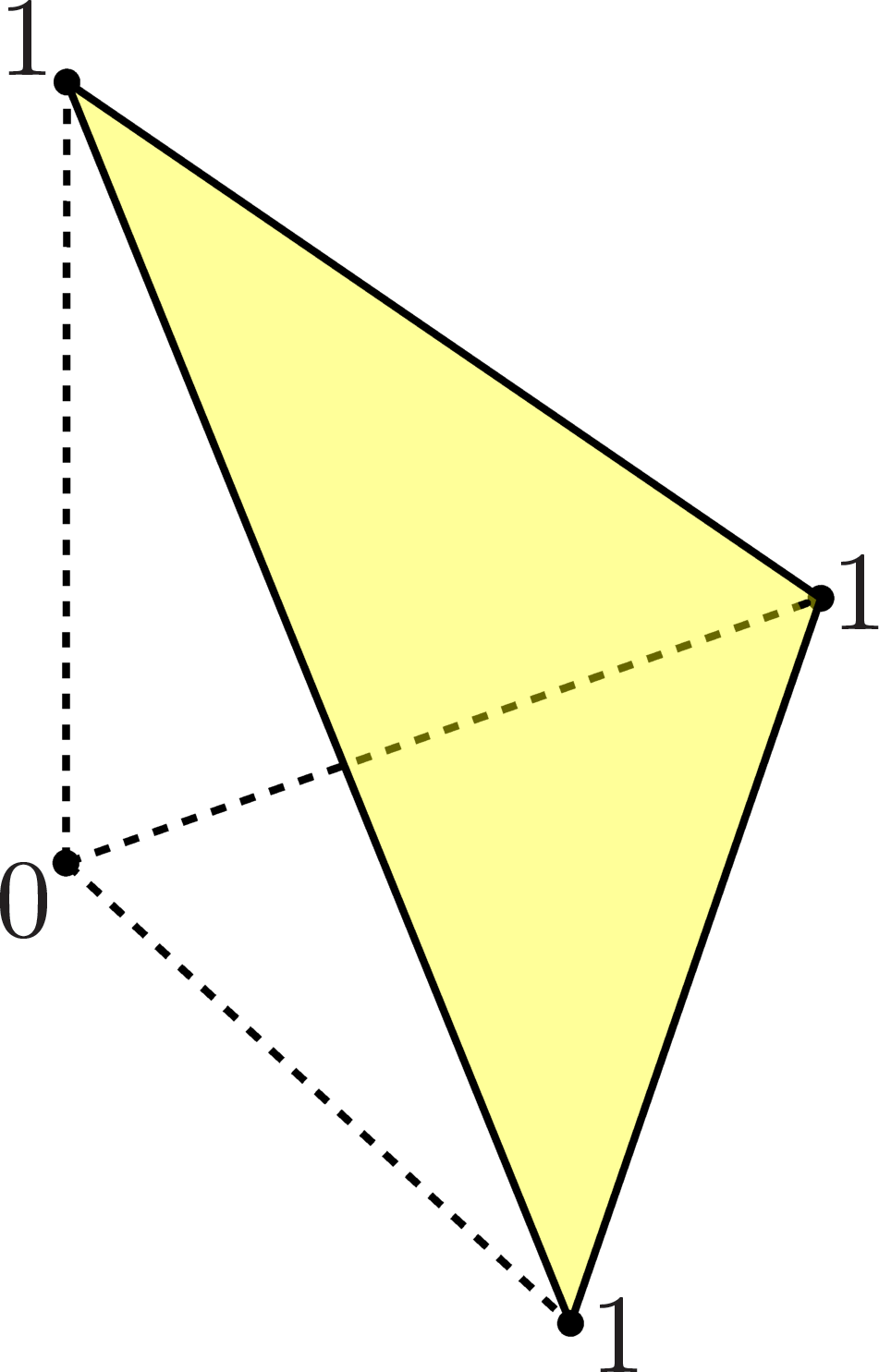}
\end{minipage} 
\hspace{0.09\textwidth}
\begin{minipage}{0.15\textwidth}
\includegraphics[width=0.9\textwidth]{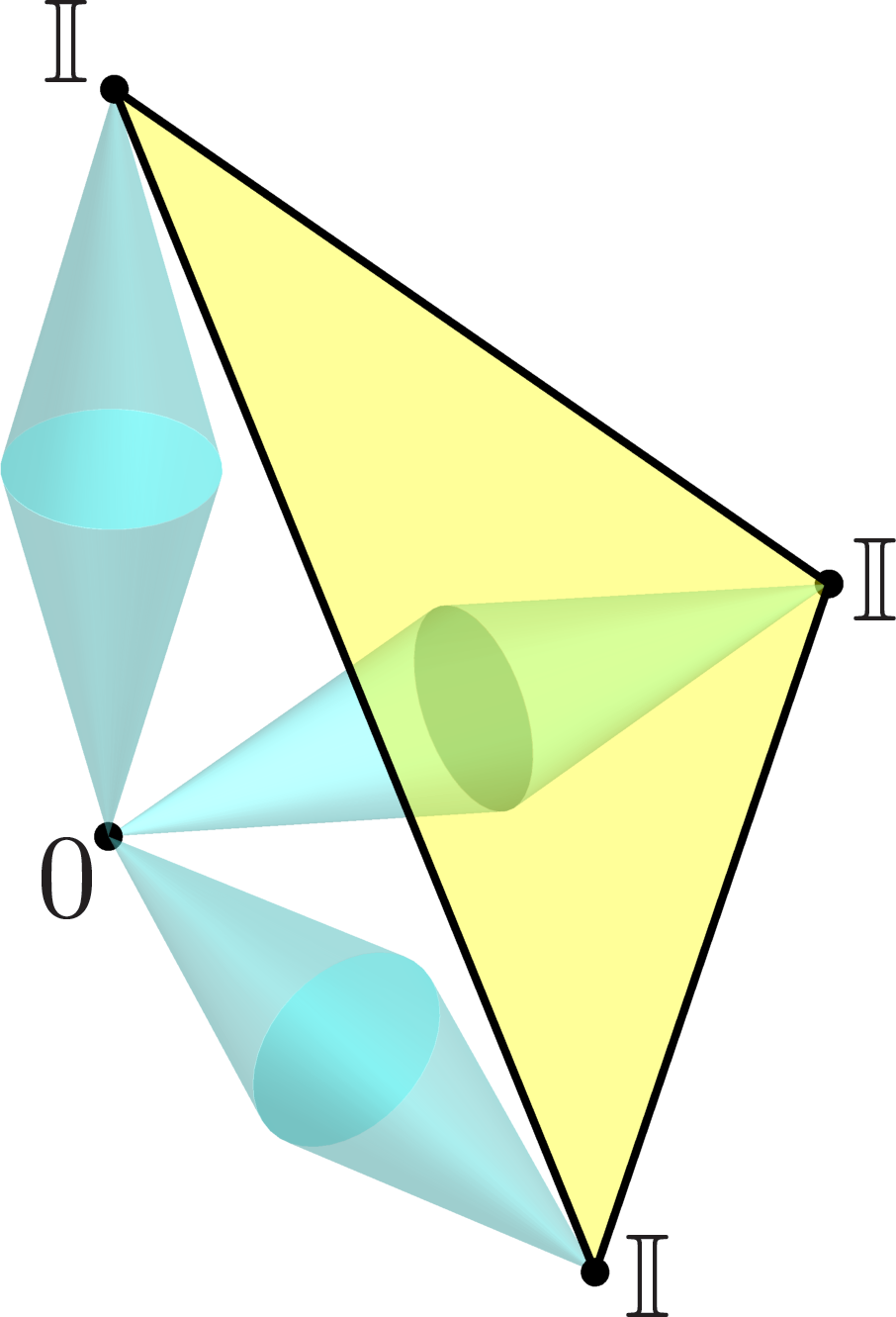}
\end{minipage} 
\caption{(Colour online) The similarity between a probability simplex and the set of 3-POVMs.} 
\label{fig:n-POVMs}
\end{figure}

While the set of POVMs $\M^n$ is perhaps unfamiliar, it is similar to the classical probability simplexes.
Figure~\ref{fig:n-POVMs} illustrates this similarity when $n=3$. The basis to construct a probability simplex is the probability range $[0,1]$. To obtain a probability simplex in $\mathbb{R}^3$, one aligns $3$ probability ranges $[0,1]$ along the $3$ axes and forms the triangle with vertices $(1,0,0)$, $(0,1,0)$, $(0,0,1)$. To construct the set of $3$-POVMs, the probability range $[0,1]$ is replaced by the set $0 \le X \le \II_A$, which for qubits forms a double cone in $\A^H$ (illustrated in Figure~\ref{fig:n-POVMs})~\cite{nguyen2016nonseparability}. One then `aligns' these $3$ sets $0 \le X \le \II_A$ along the three orthogonal component spaces of $(\A^H)^{\oplus 3}$. The set of $3$-POVMs is formed in between the points $(\II_A,0,0)$, $(0,\II_A,0)$, $(0,0,\II_A)$. This analogy between a probability simplex and the set of $n$-POVMs applies in the same way for any $n$. There is, however, a crucial difference between classical probability simplexes and sets of POVMs: while $[0,1]$ is $1$-dimensional with $2$ extreme points $0$ and $1$, the set $0 \le X \le \II_A$ is generally high-dimensional and with more extreme points other than $0$ and $\II_A$. As a result, the set of $n$-POVMs is also of high dimension and carries other extreme points apart from the \emph{special} ones at the `corners', which are of the form $\oplus_{i=1}^{n} \delta_{ik} \II_A $ with $k=1,2,\ldots n$. 

{\it The steering assemblage of $n$-POVMs.}
Now each POVM measurement $E$ performed on Alice's side gives rise to a steering ensemble on Bob's side, \mbox{$\oplus_{i=1}^{n}\Tr_{A}[\rho (E_i \otimes \II_B)]$}. This is most easily implemented by the concept of a steering function $\rho^{A \to B}: \A^H \to \B^H$, which maps $X \in \A^H$ to $\Tr_A [\rho (X \otimes \II_B)] \in \B^H$~\cite{nguyen2016nonseparability}. This induces the map $(\rho^{A \to B} )^{\oplus n}: (\A^H)^{\oplus n} \to (\B^H)^{\oplus n}$. For $X$ being an element or a subset of $\A^H$, we denote $X'= \rho^{A \to B} (X)$. The same notation is used for composite vectors, namely, for $X$ being an element or a subset of $(\A^H)^{\oplus n}$, $X'= (\rho^{A\to B})^{\oplus n} (X)$.

Geometrically, the map $(\rho^{A \to B})^{\oplus n}$ maps a point in the set of POVMs $\M^n$ to a point in $(\B^H)^{\oplus n}$. The set $(\M^n)'= (\rho^{A \to B})^{\oplus n} (\M^n)$ is called the \emph{steering assemblage of $n$-POVMs}. Being a linear image of $\M^n$, which is convex and compact~\cite{Chiribella2007a}, $(\M^n)'$ is also convex and compact. Moreover, since $\M^n$ belongs to $\P^n$, $(\M^n)'$ belongs to $(\P^n)'$. 

{\it The capacity of an ensemble of Bob's pure states.}
For an ensemble $u$ of Bob's pure states $\S_B$, the $n$-\emph{capacity} $\K^n(u)$ is the set of $n$-component ensembles that it can simulate. That is to say, the capacity $\K^n(u)$ is a subset of $(\B^H)^{\oplus n}$ consisting of composite operators $K= \oplus_{i=1}^n K_i$, each component being given by
\begin{equation}
K_i = \int \d \omega (P) u(P) G_i (P) P,
\end{equation}
with all possible choices of response functions $G_i$ that satisfy $G_i(P) \ge 0$, $\sum_{i=1}^{n} G_i(P)=1$. It is easy to show that the $n$-capacity $\K^n(u)$ is also a convex compact set, which has $n$ \emph{special extreme points} of the form $\oplus_{i=1}^n \delta_{ik} \int \d \omega (P) u(P) P$ with $k=1,2,\ldots,n$. 

{\it Steerability as a nesting problem.}
With the above definitions, the following lemma is obvious.
\begin{lemma} 
A state $\rho$ is $u$-unsteerable with $n$-POVMs if and only if $(\M^n)' \subseteq \K^n(u)$. 
\label{lem:steerable_nesting}
\end{lemma} 

We first consider the special extreme points of $\M^n$. It is easy to show that for their steering images to be in $\K^n(u)$, one has
\begin{equation}
\int \d \omega (P) u(P) P = \II_A',
\label{eq:LHS_normalisation}
\end{equation}
which is referred to as \emph{the minimal requirement} for $u$~\footnote{Since if $\int \d \omega (P) u(P) G_i(P) P = \II_A'$ for some response function $G_i$, then by taking the trace of this equality, one obtains $\int \d \omega (P) u(P) G_i(P) = 1$. Because $G_i(P) \ge 0$, this tells us that $G_i(P)=1$ almost everywhere with respect to the measure generated by $u$. This minimal requirement can also be obtained by summing both sides of equation~\eqref{eq:unsteerable_def} over $i$, bearing in mind the constraints $\sum_{i=1}^{n} E_i = \II_A$ and $\sum_{i=1}^{n} G_i (P)=1$. It can also be thought as the condition for $(\M^1)' \subseteq \K^1(u)$}.

Once reformulated in terms of a nesting problem of convex objects (Lemma \ref{lem:steerable_nesting}), one can apply nesting criteria to test steerability. The following lemma is such a nesting criterion based on a duality representation.
\begin{lemma}[Nesting criterion by duality]
Let $\X$ be a convex compact subset of a finite-dimensional Euclidean space. Then a compact subset $\Y$ is contained in $\X$ if and only if $ \max_{X \in \X} \dprod{Z}{X} \ge \max_{Y \in \Y} \dprod{Z}{Y}$ for all vectors $Z$ in the space. 
\label{lem:first_nesting}
\end{lemma}
The idea behind this lemma is that if $\X$ contains $\Y$, then its projection onto any direction contains that of $\Y$  and vice versa (see Figure~\ref{fig:nesting_duality}). A full proof is given in Appendix~\ref{sec:nesting_by_duality}.

\begin{figure}[!h]
    \begin{tikzpicture}
    \node at (0,0) {\includegraphics[width=0.2\textwidth]{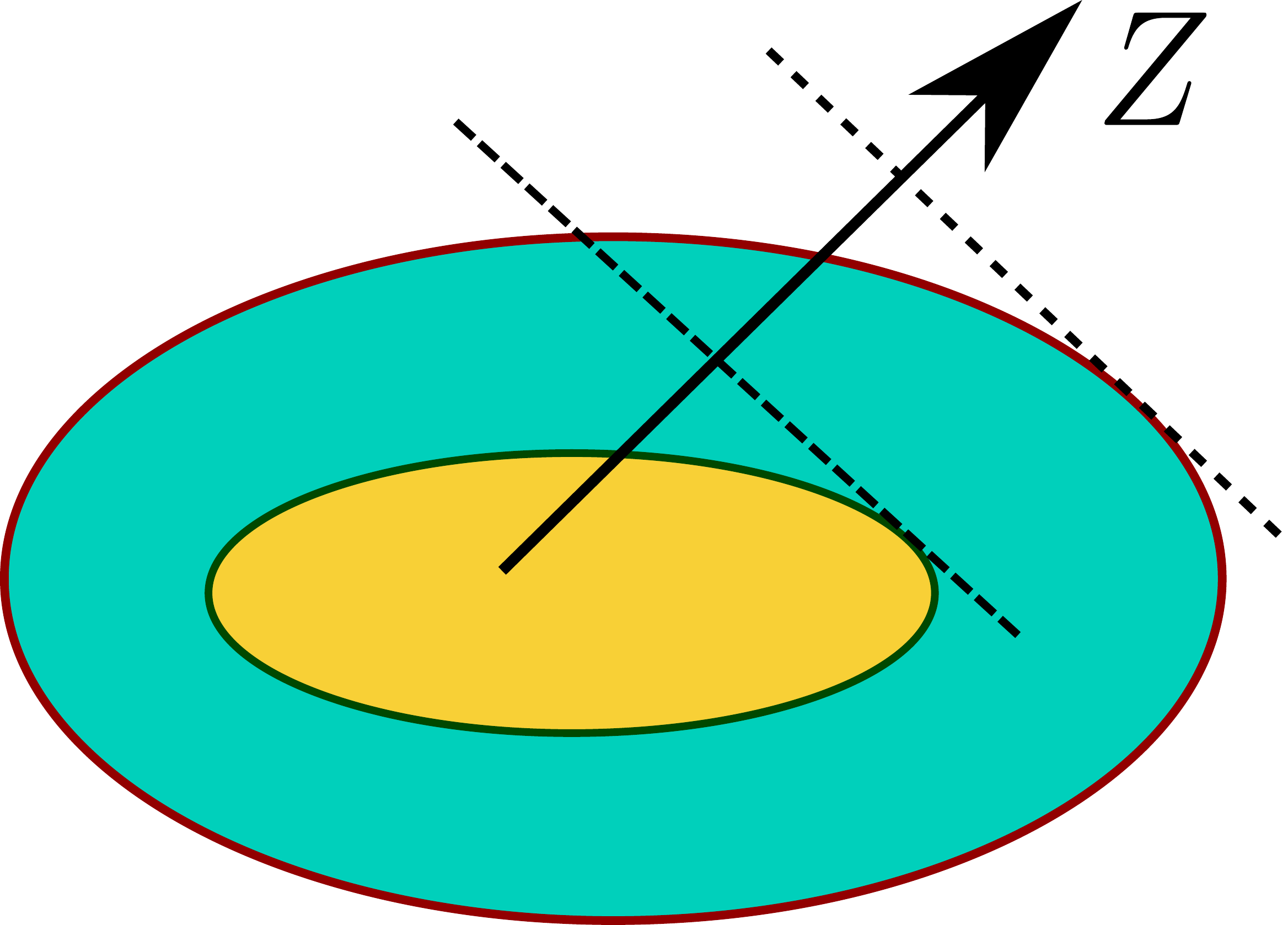}};
    \node at (0,-0.35) {$\Y$};
    \node at (0.25,-1) {$\X$};
    \end{tikzpicture}
    \caption{(Colour online) Nesting by duality.}
    \label{fig:nesting_duality}
\end{figure}

To apply this lemma with $\X = \K^n(u)$, $\Y=(\M^n)'$, we need to solve two maximisation problems: $\max_{K \in \K^n(u)} \dprod{Z}{K}$ and $\max_{E \in \M^n} \dprod{Z}{E'}$ for a given composite operator $Z$ in $(\B^H)^{\oplus n}$. While the latter is a semidefinite program, the former is a linear constrained maximisation, which can be solved explicitly (see Appendix~\ref{sec:central_maximisation} for the details):
\begin{equation}
\max_{K \in \K^n(u)} \dprod{Z}{K}  = \int \d \omega (P) u(P) \max_i \dprod{Z_i}{P}.
\label{eq:first_maximisation_solution}
\end{equation} 

From Lemma~\ref{lem:steerable_nesting} and Lemma~\ref{lem:first_nesting}, the following theorem immediately follows.

\begin{theorem} 
A bipartite state $\rho$ is $u$-unsteerable if and only if
\begin{equation}
\int \d \omega (P) u(P) \max_i \dprod{Z_i}{P} \ge \max_{E \in \M^n} \sum_{i=1}^{n} \dprod{Z_i}{E_i'}
\label{eq:ineq_main}
\end{equation}
for all composite operators $Z= \oplus_{i=1}^{n} Z_i$ in $(\B^H)^{\oplus n}$.
\label{th:main}
\end{theorem}

Inequality~\eqref{eq:ineq_main} is the main result in this Letter: it is valid for systems of arbitrary dimension and POVMs of arbitrary number of outcomes. We now discuss some of its important consequences; details of proofs and further discussions are given in Appendix~\ref{sec:corollaries_main}.  

{Had one started with the original definition of quantum steering with an ``indexed LHS ensemble'', that is, an LHS ensemble indexed by some hidden variable, one would arrive at a similar inequality as~\eqref{eq:ineq_main}. In that case, the integration is taken over the hidden variable instead (see Appendix~\ref{sec:corollaries_main}). However, one can rewrite it as an integral over the push-forward measure over Bob's pure states $\S_B$~\cite{Bogachev2007a}. This implies that it is only the push-forward measure on $\S_B$ that really determines the capacity of an LHS ensemble. Two indexed LHS ensembles generating the same measure on Bob's pure states would have the same capacity. In other words, our definition of quantum steering where the local hidden variable is omitted is equivalent to the original definition of quantum steering (Corollary 1, Appendix~\ref{sec:corollaries_main}). Having eliminated the arbitrary choice of local hidden variable in quantum steering, the symmetry of the state directly has a stronger implication on the symmetry of LHS ensembles (see Theorem~\ref{th:second}). In fact, it is this stronger implication of symmetry that actually renders many unsolvable problems in Bell nonlocality solvable in the context of quantum steering~\cite{wiseman2007steering}.


More specifically, the state $\rho$ is said to have $(\G,U,V)$-symmetry with $\G$ being a group and $U$ and $V$ being its two representations on $\H_A$ and $\H_B$, respectively, if $\rho= U^{\dagger} (g) \otimes V^{\dagger}(g) \rho U (g) \otimes V(g)$ for all $g \in \G$. The action $V$ of $\G$ on Bob's pure states $\S_B$ generates an action $R_V$ on the space of distributions on $\S_B$ defined by $[R_V(g) u] (P) = u[V^\dagger(g) P V(g)]$. 
We then have a strengthened form of Lemma 1 of Ref.~\cite{wiseman2007steering} on the symmetry of LHS ensemble. 
\begin{theorem}[Symmetry of LHS ensembles]
For a given state $\rho$ which is $(\G,U,V)$-symmetric with a compact group $\G$, if $\rho$ is unsteerable with $n$-POVMs then it admits an LHS ensemble $u^{\ast}$ which is $(\G,R_V)$-invariant, i.e., $u^{\ast}=R_V(g) u^{\ast}$ for all $g$ in $\G$.
\label{th:second}
\end{theorem}
This theorem is applicable as well when measurements are restricted to PVMs. To understand the difference with Lemma 1 of Ref.~\cite{wiseman2007steering}, we consider the example where $\G$ acts transitively on Bob's pure states $\S_B$ (i.e., a single orbit covers all of $\S_B$).  If $\rho$ is unsteerable, Lemma 1 of Ref.~\cite{wiseman2007steering} then states the existence of an indexed LHS ensemble, on which $\G$ acts covariantly on the indices and the states. However, due to the arbitrariness in choice of the local hidden variable, there exist in fact infinitely many different $\G$-covariant indexed LHS ensembles (see Appendix~\ref{sec:corollaries_main}). One then could not single out an unique choice of LHS ensemble. On the other hand, under the same conditions, Theorem~\ref{th:second} implies that the state is unsteerable with the unique uniform distribution on Bob's pure states as an LHS ensemble.
}


Beyond revealing very general aspects of quantum steering, Theorem~\ref{th:main} can also be used to test steerability in practice. 
The most difficult part is to determine the existence of a LHS ensemble $u$. 
Even when measurements are limited to PVMs, the question is so far solved only for highly symmetric states, e.g., the Werner state~\cite{wiseman2007steering} and the two-qubit $T$-state (mixtures of Bell states)~\cite{jevtic2015einstein,nguyen2016necessary}. 
The implication of our approach on this problem will be discussed elsewhere. 
However, as we mentioned, even when $u$ is known, the problem of determining steerability with POVMs is still open~\cite[Problem 39]{oqp}.
It is this latter problem that we are concerned with in the following. 
We will show that Theorem~\ref{th:main} can provide a strong numerical evidence for steering with POVMs with a given identified candidate for the LHS ensemble. 

{\it The gap function.}
We first note that for a LHS ensemble $u$ satisfying the minimal requirement~\eqref{eq:LHS_normalisation}, the inequality~\eqref{eq:ineq_main} is invariant with respect to the transformation $\oplus_{i=1}^{n} Z_i \to \frac{1}{\sqrt{D}} \oplus_{i=1}^{n} (Z_i-C)$, where $C=\frac{1}{n}\sum_{i=1}^{n} Z_i$ and $D= \sum_{i=1}^{n} \dprod{Z_i-C}{Z_i-C}$. We can therefore restrict $Z$ to the set of those satisfying $\sum_{i=1}^{n} Z_i= 0$ and $\sum_{i=1}^{n} \dprod{Z_i}{Z_i}=1$, denoted by $\C^n$.
For clarity, we introduce the \emph{gap function} $\Delta [(\M^n)',\K^n(u)]$, defined to be 
\begin{align}
\min_{Z \in \C^n} \left\{ \int \d \omega (P) u(P) \max_i \dprod{Z_i}{P} - \max_{E \in \M^n} \sum_{i=1}^{n} \dprod{Z_i}{E_i'} \right\}.
\label{eq:first_criterion}
\end{align}
The gap function characterises the gap between the boundary of $(\M^n)'$ and that of $\K^n(u)$ from inside. The state $\rho$ is $u$-unsteerable with $n$-POVMs if and only if $\Delta [(\M^n)',\K^n(u)] \ge 0$.  

{\it Restriction to rank-$1$ POVMs.}
As all POVMs can be post-processed from those of rank-$1$, to test quantum steerability we can concentrate on the latter~\cite{barrett2002}.  To this end, we define $\mathcal{\widetilde{\M}}^n=\{ \oplus_{i=1}^{n} \alpha_i P_i\}$ where $P_i$ are (not necessarily independent) rank-$1$ projections and $0 \le \alpha_i \le 1$ such that $\sum_{i=1}^{n} \alpha_i P_i = \II_A$. To test the steerability with $\M^n$, we therefore only need to calculate $\Delta [(\widetilde{\M}^n)',\K^n(u)]$.

{\it Steerability of two-qubit Werner states with POVMs.}
Consider the two-qubit Werner state,
\begin{equation}
W_p= p \ketbra{\psi^-}{\psi^-} + (1-p) \frac{\II_A}{2} \otimes \frac{\II_B}{2},
\label{eq:werner}
\end{equation}
which is a mixing between the singlet Bell state $\ket{\psi^-} = \frac{1}{\sqrt{2}} (\ket{01}-\ket{10})$ and the maximally mixed state with mixing parameter $p$ ($0 \le p \le 1$). When restricted to PVMs, by explicitly constructing the response functions, it was shown that the Werner state is unsteerable for the mixing probabilities $p \le \frac{1}{2}$~\cite{wiseman2007steering}. By adapting Barrett's model~\cite{barrett2002} of local hidden variables, it was further possible to show that Werner states are unsteerable with POVMs when $p \le \frac{5}{12}$~\cite{quintino2015wernerlhsm}. We are to study the conjecture~\cite{werner2014steering}:
\begin{conjecture}
The Werner state with mixing parameter $p \le \frac{1}{2}$ is unsteerable for all $n$-POVMs. That is to say, POVMs and PVMs are equivalent for steering two-qubit Werner states.
\label{conj:werner_steering}
\end{conjecture}
Although further analyses restricted to finite subsets of POVMs~\cite{cavalcanti2016lhsmalg,hirsch2016lhsmalg} or POVMs with special symmetry~\cite{werner2014steering} support unsteerability of the Werner state for $\frac{5}{12} \le p \le \frac{1}{2}$~\cite{bavaresco2017finitepovms,werner2014steering}, there has not been a concrete evidence when one considers all POVMs. For $n=1$, the conjecture is obvious. For $n=2$, it has been proven by demonstrating the equivalence to steering with PVMs~\cite{nguyen2016nonseparability}. The proof for $n=3$ is also known~\cite{werner2014steering}. Finally, it is known that it is sufficient to consider the conjecture for $n = 4$~\cite{d2005classical,werner2014steering}. Here, by computing the corresponding gap function for the Werner state, we provide strong numerical evidence for Conjecture~\ref{conj:werner_steering} for $n=4$.  

\begin{figure}[!t]
	\begin{center}
		\includegraphics[width=0.45\textwidth]{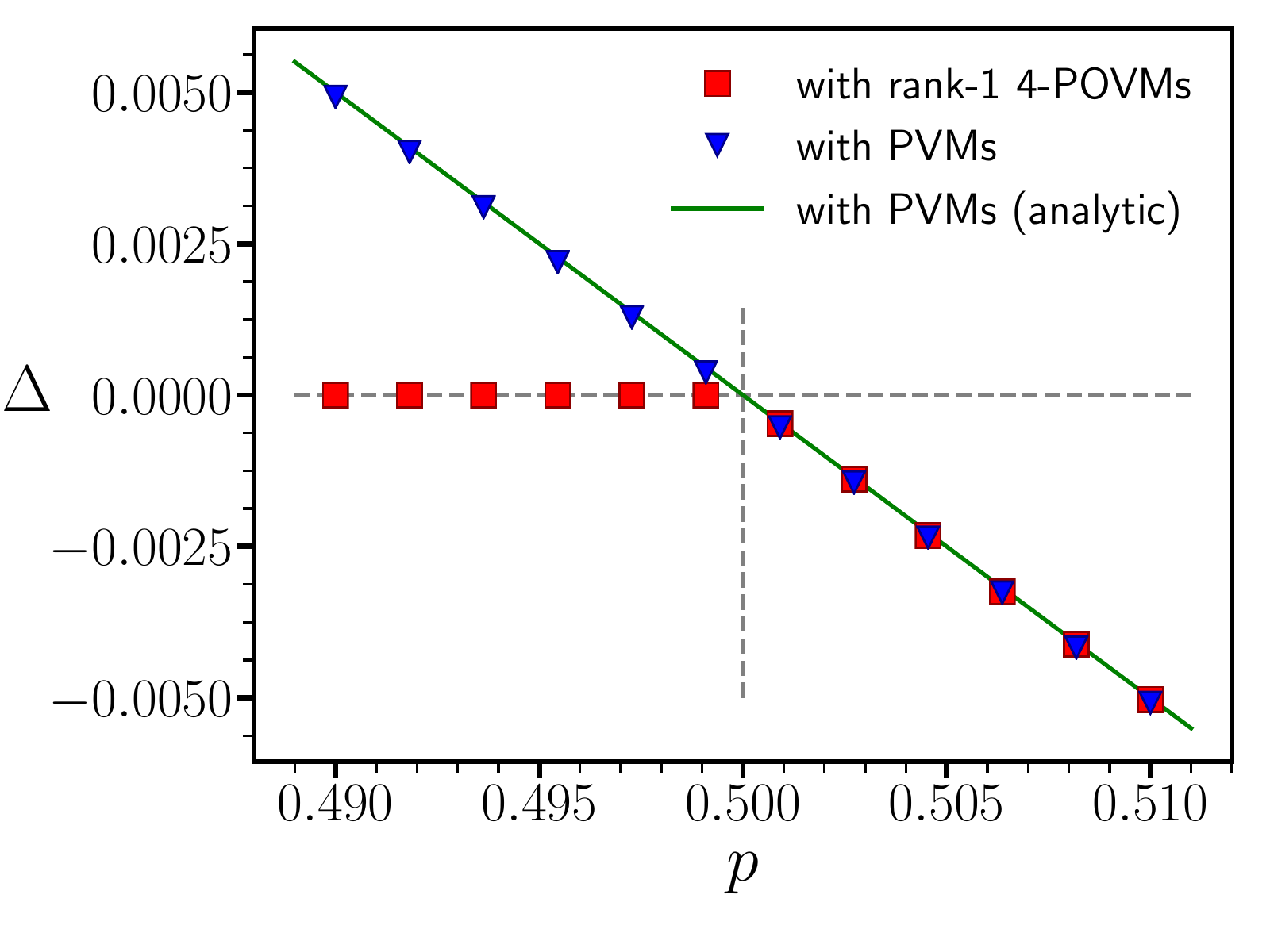}\\
	\end{center}
	\caption{(Colour online) The gap function for the two-qubit Werner state with mixing probability around $\frac{1}{2}$ for rank-1 4-POVMs and for PVMs.}
	\label{fig:werner}
\end{figure}

It is easy to see that the Werner state has $\operatorname{U}(2)$ symmetry as defined above. Moreover, since the action is transitive on Bob's Bloch sphere, by Theorem~\ref{th:second}, the candidate for the LHS ensemble can be limited to the uniform distribution, $u=1$. 
To simplify the notation, from now on we simply use $\Delta$ to denote the gap function without specifying the set of measurements and the capacity, which can be understood from the context. The computation of the gap function $\Delta$ generally requires global minimisation over $Z$, which is carried out by the standard (non-deterministic) simulated annealing algorithm (see Appendix~\ref{sec:simulated_annealing} for details).

In Figure~\ref{fig:werner}, we present values of $\Delta$ found for the mixing probability $p \approx \frac{1}{2}$. For comparison, we also present the numerical results of $\Delta$ when the measurements are limited to PVMs, which are in very good agreement with the analytical calculation (see Appendix~\ref{sec:PVM_gap}). For $p > \frac{1}{2}$, one observes that the gap function $\Delta$ for rank-$1$ $4$-POVMs is negative and coincides with the gap function for PVMs. For all $p \le \frac{1}{2}-10^{-3}$, the best obtained values for $\Delta$ are less than $10^{-10}$ but persistently non-negative. The fact that the gap function tends to vanish for $p \le \frac{1}{2}$ instead of attending finite positive values seems to be because of the high dimensionality of POVMs~\footnote{The effect of dimensionality can be understood directly from Figure~\ref{fig:nesting_duality}. On the plane, the gap between the boundaries of the two objects is positive. However when imbedded in the 3D space, the gap will be zero in direction that is perpendicular to plane that contains the two objects.}
The presented results therefore support conjecture~\ref{conj:werner_steering}. Unfortunately, within this work, the ambiguity region of $\frac{1}{2} - 10^{-3} \le p \le \frac{1}{2}$ cannot be resolved due to limits of numerical accuracy (see Appendix~\ref{sec:simulated_annealing}). 
 
{\it Conclusion.}
Our work simplifies the definition of quantum steering, where we show that the local hidden variable indexing LHS ensemble can be omitted. As a direct consequence, a stronger theorem on the symmetry of the LHS ensemble is derived. We have thereby opened a general approach to studying quantum steerability with POVMs. Further works to strengthen the numerical evidence for the unsteerability of the Werner state at $p=\frac{1}{2}$ and testing steerability with POVMs of arbitrary two-qubit states are underway. Moreover, although the current illustrative applications are based on two qubits, our approach is not limited by the dimensionality of the systems. It is hoped that systematic tests for steering with POVMs (particularly for high dimensional systems) will give a complete answer to the fundamental question of the equivalence between POVMs and PVMs for steerability. Beyond quantum steering, we leave open the question of whether this approach can be extended to characterise Bell nonlocality with POVMs.  
\begin{acknowledgements} 
We thank Jessica Bavaresco, Johannes Berg, Anna C. S. Costa, Otfried G\"uhne, Marcus Huber, Michael Jarret, Brad Lackey, X. Thanh Le, N. Duc Le, H. Viet Nguyen, Matthew Nicol, Ana Bel\'{e}n Sainz, Roope Uola, Reinhard Werner, and Howard Wiseman for useful discussions on general and technical aspects of this work. SJ is supported by an Imperial College London Junior Research Fellowship. CN acknowledges the support from Erwin Schr\"odinger International Institute for Mathematics and Physics (ESI) during his participation in the program Quantum Physics and Gravity.
\end{acknowledgements}
\appendix
\section{The compactness of the capacity}
\label{sec:compact_K}
In this appendix, we prove the compactness of the capacity $\K^n(u)$. On the Bloch sphere $\S_B$, the distribution $u$ defines a measure $\mu$. We assume that the response functions $G_i(P)$ are squared-integrable with respect to $\mu$. Consider the space $\L^2(\S_B,\mu)$ of squared-integrable functions  on $\S_B$ with respect to $\mu$, which is a Hilbert space (as usual, we ignore the difference on a zero-measured set)~\cite[Chapter V]{Hu1967a}. We then construct the Hilbert space $[\L^2(\S_B,\mu)]^{\oplus n}$ as usual.  We define the subset of $\Omega= \{G= \oplus_{i=1}^{n} G_i  \in [\L^2(\S_B,\mu)]^{\oplus n} : 0 \le G_i(P) \le 1, \sum_{i=1}^{n} G_i(P)= 1 \}$, which is closed and convex, thus weakly closed~\cite[Chapter V, Corollary 1.5]{Conway1990a}. Moreover, it is obviously bounded, thus weakly compact~\cite[Chapter V, Theorem 4.2]{Conway1990a}. 

Now consider the linear operator $T:[\L^2(\S_B,\mu)]^{\oplus n} \to (\B^H)^{\oplus n}$, defined by $T(G)= \oplus_{i=1}^{n} \int \d \mu (P) G_i (P) P$. It is obvious that $T$ is bounded, thus continuous and weakly continuous~\cite[Chapter VI, Theorem 1.1]{Conway1990a}.  It then follows directly that $\K^n(u)=T(\Omega)$ is compact.
\section{Nesting criterion by duality}
\label{sec:nesting_by_duality}
In this Appendix, we provide the proof for the nesting criterion by duality.
\setcounter{lemma}{1}
\begin{lemma}[Nesting criterion by duality]
Let $\X$ be a convex compact subset of a finite-dimensional Euclidean space. Then a compact subset $\Y$ is contained in $\X$ if and only if $ \max_{X \in \X} \dprod{Z}{X} \ge \max_{Y \in \Y} \dprod{Z}{Y}$ for all vectors $Z$ in the space. 
\end{lemma}
\begin{proof}
It is obvious that if $\Y \subseteq \X$ then $\max_{X \in \X} \dprod{Z}{X} \ge \max_{Y \in \Y} \dprod{Z}{Y}$ for all $Z$. Now suppose  $ \max_{X \in \X} \dprod{Z}{X} \ge \max_{Y \in \Y} \dprod{Z}{Y}$ for all $Z$ and $\Y \not\subseteq \X$. Because  $\Y \not\subseteq \X$, there exists $A \in \Y$, $A \not\in \X$. Since $\X$ is a convex and compact set, by the separation theorem, $A$ is separated from $\X$ by a hyperplane, i.e., there exists a vector $Z$ such that $\dprod{Z}{A} > \max_{X \in \X} \dprod{Z}{X}$~\cite{rockafellar2015convex}. It follows that $ \max_{Y \in \Y} \dprod{Z}{Y}  \ge \dprod{Z}{A} > \max_{X \in \X} \dprod{Z}{X}$, contradicting the assumption. 
\end{proof}
\section{Solving the first optimisation problem}
\label{sec:central_maximisation}
Here we provide details of the solution to the constrained maximisation problem~\eqref{eq:first_maximisation_solution}. Using the definition of $\K^n(u)$, we have
\begin{equation}
\max_{K \in \K^n(u)} \dprod{Z}{K} = \max_{G} \int \d \omega (P) u (P) \sum_{i=1}^{n} G_i(P) \dprod{Z_i}{P},
\label{eq:first_maximisation}
\end{equation}
subject to the constraints $G_i(P) \ge 0$ and $\sum_{i=1}^{n} G_i(P)=1$. This is a linear maximisation problem with linear constraints, which can be solved easily by Lagrange's multipliers.
For every constraint $\sum_{i=1}^{n} G_i(P) = 1$ for each $P \in \S_A$, we introduce a Lagrange's multiplier $\lambda (P)$.
This leads us to a modified unconstrained maximisation problem
\begin{align}
I[\lambda (P)] &=   \max_{G} \left\{ \int \d \omega (P) u (P) \sum_{i=1}^{n} G_i(P) \dprod{Z_i}{P} \right. \nonumber \\
& \qquad \left. - \int \d \omega (P) \lambda (P) \left[\sum_{i=1}^{n} G_i(P) - 1 \right] \right\} \nonumber \\
 &=   \max_{G} \left\{ \int \d \omega (P)  \sum_{i=1}^{n} G_i(P) \left[ u(P) \dprod{Z_i}{P} -\lambda(P)\right]  \right\} \nonumber \\
& \qquad + \int \d \omega (P) \lambda (P).  
\end{align} 
Note that the function under maximisation is a linear function of $G_i(P)$, which is bounded by $0\le G_i(P)\le 1$. Therefore $I[\lambda (P)]$ is saturated by
\begin{equation}
G_i^\ast (P) = \Theta [u(P) \dprod{Z_i}{P} - \lambda (P)],
\label{eq:solution_lambda}
\end{equation}
where $\Theta$ is the Heaviside step function, $\Theta (x) = 1$ if $x \ge 0$, and $\Theta (x) = 0$ otherwise. One now needs to choose $\lambda (P)$ such that the constraint is satisfied,
\begin{equation}
\sum_{i=1}^{n} \Theta [u(P) \dprod{Z_i}{P} - \lambda (P)] = 1.
\end{equation}
Consider some fixed $P$. The last equation means that $\lambda (P)$ must be such that out of $\{u(P) \dprod{Z_i}{P}\}_{i=1}^{n}$, only one is larger than or equal to $\lambda (P)$. In other words, the suitable choice for $\lambda (P)$ is
\begin{equation}
\lambda (P) = u(P) \max_i \dprod{Z_i}{P}.
\end{equation}
With this solution for $\lambda (P)$, substituting~\eqref{eq:solution_lambda} to~\eqref{eq:first_maximisation} one then obtains~\eqref{eq:first_maximisation_solution}. 

So far we actually ignored the case where for some $P$, $\max_i \dprod{Z_i}{P}$ is attained by two indices, say, $i=i_1$ and $i=i_2$. In this case, one then has to slightly modify~\eqref{eq:solution_lambda}: at such a point $P$, while $G_i^{\ast} (P) = 0$ for $i \ne i_1, i_2$,  $G_{i_1}^{\ast}(P)$ and $G_{i_2}^{\ast} (P)$ can take arbitrary values between $0$ and $1$, provided that $G_{i_1}^{\ast} (P) + G_{i_2}^{\ast} (P) = 1$. Similar modification is needed if $\max_i \dprod{Z_i}{P}$ is attained by more indices. The maximal value~\eqref{eq:first_maximisation_solution} however remains the same.
\section{Corollaries of Theorem~\ref{th:main}}
\label{sec:corollaries_main}
\subsection{The equivalence between our definition and the original definition of steering}
\label{sec:equivalent_definitions}
The original definition of quantum steering~\cite{wiseman2007steering} goes as follows. Let $(\Lambda,\nu)$ be a probability measure space (we ignore the symbol which denotes the $\sigma$-algebra for the measure $\nu$). Let $F:\Lambda \to \S_B$ be a measurable function from the index space $(\Lambda,\nu)$ to the set of Bob's pure states $\mathcal{S}_B$. A state $\rho$ is then called  $(\Lambda,\nu)$-\emph{unsteerable} (from Alice's side) with respect to $n$-POVMs if, for any $n$-POVM $E$, Alice can find $n$ response functions $G_i$ (with $G_i(\lambda) \ge 0$ and $\sum_{i=1}^{n} G_i(\lambda)= 1$ for all $\lambda \in \Lambda$) such that the steering ensemble can be simulated via a \emph{local hidden state model}~\cite{wiseman2007steering},
\begin{equation}
\Tr_{A}[\rho (E_i \otimes \II_B)]= \int_{\Lambda} \d \nu (\lambda) \, \,  G_i(\lambda) F(\lambda),
\label{eq:indexed_unsteerable_def}
\end{equation}
where the integral is taken over the index space $\Lambda$. In this case, we say $\rho$ admits an \emph{indexed LHS model}.

If $\rho$ satisfies the definition of steering in the main text, it is cleared that it admits an indexed LHS model, where the index space is Bob's pure states themselves. On the other hand, if $\rho$ admits an indexed LHS model, it is \emph{unclear} that the existence of a response function on Bob's pure states is guaranteed. This is particularly important if the indexing function $F$ is a many-to-one function. One of the strengths of our approach is that it provides a proof that a response function on Bob's pure states does exist. We know of no other (constructive) proof at the moment. 
\begin{corollary}
Our definition of quantum steering is equivalent to the conventional definition of quantum steering where the LHS ensemble is indexed by a hidden variable.
\label{cor:1}
\end{corollary}
\begin{proof}
Suppose $\rho$ admits an indexed LHS ensemble, which is indexed by $(\Lambda,\nu)$. Then following the argument that leads to Theorem 1 in the main text, we arrive at the following statement: 
\begin{lemma*} 
A bipartite state $\rho$ is $(\Lambda,\nu)$-unsteerable if and only if
\begin{equation}
\int_{\Lambda} \d \nu (\lambda) \max_i \dprod{Z_i}{F(\lambda)} \ge \max_{E \in \M^n} \sum_{i=1}^{n} \dprod{Z_i}{E_i'}
\label{eq: indexed_ineq_main}
\end{equation}
for all composite operators $Z= \oplus_{i}^{n} Z_i$ in $(\B^H)^{\oplus n}$.
\end{lemma*}
Now denote by $\mu$ the push-forward measure generated by $F$ on the set of Bob's pure states~\cite{Bogachev2007a}. Changing the variable in the integral, one has 
\begin{equation}
\int_{\Lambda} \d \nu (\lambda) \max_i \dprod{Z_i}{F(\lambda)} = \int \d \mu (P) \max_i \dprod{Z_i}{P},
\end{equation} 
where the latter integral is taken over Bob's pure states.
Let $u$ be the distribution on Bob's pure states generated by $\mu$ with respect the Haar measure $\omega$. This means the inequality~\eqref{eq:ineq_main} in the main text is satisfied for distribution $u$. According to Theorem~\ref{th:main}, $\rho$ is $u$-unsteerable. 
\end{proof}

\subsection{The symmetry of LHS ensembles}
For a given state $\rho$, we denote by $\Omega^n(\rho)$ the set of ensembles $u$ over Bob's pure states such that $\rho$ is $u$-unsteerable with $n$-POVMs ($\Omega^n(\rho)$ is empty if $\rho$ is steerable).  From inequality~\eqref{eq:ineq_main}, it is easy to see that $\Omega^n(\rho)$ is convex (Corollary~\ref{cor:2}). Moreover, we show that the symmetry of $\rho$ implies the symmetry of $\Omega^n(\rho)$. 
\begin{corollary}
For a given state $\rho$, $\Omega^{n}(\rho)$ is convex.
\label{cor:2}
\end{corollary}
\begin{proof}
Suppose $u_1$ and $u_2$ are in $\Omega^{n}(\rho)$. That is to say, $u_1$ and $u_2$ satisfy inequality~\eqref{eq:ineq_main}. It is then easy to check that inequality~\eqref{eq:ineq_main} is also satisfied for all convex combinations of $u_1$ and $u_2$. In other words, all convex combinations of $u_1$ and $u_2$ are in $\Omega^{n}(\rho)$.
\end{proof}
Let us recall from the main text that the state $\rho$ is said to have $(\G,U,V)$-symmetry with $\G$ being a group and $U$ and $V$ being its two representations on $\H_A$ and $\H_B$, respectively, if $\rho= U^{\dagger} (g) \otimes V^{\dagger}(g) \rho U (g) \otimes V(g)$ for all $g \in \G$. The action $V$ of $\G$ on Bob's pure states $\S_B$ generates an action $R_V$ on the space of distributions on $\S_B$ defined by $[R_V(g) u] (P) = u[V^\dagger(g) P V(g)]$. 
\begin{corollary}
For a given state $\rho$ which is $(\G,U,V)$-symmetric, then $\Omega^n(\rho)$ is $(\G,V)$-symmetric, i.e., $\Omega^n(\rho)= R_V (g)\Omega^n(\rho)$ for all $g \in \G$.
\label{cor:3}
\end{corollary}
\begin{proof}
We need to show that if $\rho$ is $u$-unsteerable for some $u$, that is, if inequality~\eqref{eq:ineq_main} holds for $u$, then it also holds for $R_V(g)u$. Due to the symmetry of $\rho$, we have
\begin{widetext}
\begin{align}
\max_{E \in \M^n} \sum_{i=1}^{n} \Tr [ \rho (E_i \otimes Z_i)] &= \max_{E \in \M^n} \sum_{i=1}^{n} \Tr [U^{\dagger}(g) \otimes V^\dagger(g) \rho U(g) \otimes V(g) (E_i \otimes Z_i)] \nonumber \\
&= \max_{E \in \M^n} \sum_{i=1}^{n} \Tr [\rho  \ U(g) E_i U^{\dagger}(g) \otimes V(g) Z_i V^{\dagger}(g)] \nonumber \\
&= \max_{E \in \M^n} \sum_{i=1}^{n} \Tr [\rho  \ E_i \otimes V(g) Z_i V^{\dagger}(g)],
\end{align}
\end{widetext}
where the last equality is because $\M^n$ is symmetric under the action $U$ of $\G$. Then inequality~\eqref{eq:ineq_main} is equivalent to
\begin{equation}
\int \d \omega (P) u(P) \max_i \dprod{Z_i}{P} \ge \max_{E \in \M^n} \sum_{i=1}^{n} \dprod{V(g) Z_i V^{\dagger}(g)}{E_i'}.
\end{equation}
That this inequality holds for all $Z$ is then equivalent to
\begin{equation}
\int \d \omega (P) u(P) \max_i \dprod{V^{\dagger}(g) Z_i V(g)}{P} \ge \max_{E \in \M^n} \sum_{i=1}^{n} \dprod{Z_i}{E_i'}
\end{equation}
for all $Z$. Now we manipulate the left-hand side, 
\begin{widetext}
\begin{align}
\int \d \omega (P) u(P) \max_i \dprod{V^{\dagger}(g) Z_i V(g)}{P} &= \int \d \omega (P) u(P) \max_i \dprod{Z_i}{V(g)PV^{\dagger}(g)} \nonumber \\
&= \int \d \omega (P) u[V^{\dagger}(g)PV(g)] \max_i \dprod{Z_i}{P},
\end{align}
\end{widetext}
where the last inequality is a change of integration variable. By definition, $R_V(g)u(P)= u(V^{\dagger}(g)PV(g))$. Therefore, inequality~\eqref{eq:ineq_main} indeed holds for $R_V(g)u$.
\end{proof}

\setcounter{theorem}{1}
\begin{theorem}[Symmetry of LHS ensemble]
For a given state $\rho$ which is $(\G,U,V)$-symmetric with a compact group $\G$, if $\rho$ is unsteerable with $n$-POVMs then it admits an LHS ensemble $u^{\ast}$ which is $(\G,R_V)$-invariant, i.e., $u^{\ast}=R_V(g) u^{\ast}$ for all $g$ in $\G$.
\end{theorem}

\begin{proof}[Proof 1]
Since $\rho$ is unsteerable, there exists $u \in \Omega^n(\rho)$. By Corollary~\ref{cor:3}, $R_V(g) u \in \Omega^n(\rho)$ for all $g \in \G$. Since $\Omega^n(\rho)$ is convex (Corollary 2), the average over the Haar measure of $\G$, i.e., $u^{\ast}=\int_\G \d \mu(g) R_V(g) u \in \Omega^n(\rho)$, also belongs to $\Omega^n(\rho)$. This averaged distribution $u^{\ast}$ is obviously invariant under the action of $\G$, i.e., $u^{\ast}=R_V(g) u^{\ast}$ for all $g$ in $\G$.
\end{proof}
For completeness, we also provide an alternative proof of this theorem without the use of inequality~(\ref{eq:ineq_main}). This proof is more similar to the original proof in Ref.~\cite{wiseman2007steering}; to get the stronger statement, one has to apply the so-called mean value theorem for integrals~\cite{Bogachev2007a}, though.
\begin{proof}[Proof 2]
Suppose $\rho$ is $u$-unsteerable with $n$-POVMs, then for a POVM $E$, there exists response function $G$ such that
\begin{equation}
E'_i= \int \d \omega (P) u (P) G_i^E (P) P.
\end{equation}
Here, to track the dependence of the response function $G$ on the measurement, we introduce the superscript $E$ for $G$.
Now due to the symmetry of the state, we also have
\begin{equation}
E'_i= \int \d \omega (P) R_V (g) u (P) G_i^{U^\dagger(g) E U(g)} [V^{\dagger} (g) P V(g)] P.
\end{equation}
Since the right-hand-side is independent of $g$, we can take the average over $g$ with respect to the Haar measure $\mu$ of $\G$,
\begin{equation}
E'_i= \int \d \omega (P) P \int d\mu (g) R_V(g) u (P) G_i^{U^\dagger(g) E U(g)} [V^{\dagger} (g) P V(g)]. 
\end{equation}
According to the mean value theorem~\cite{Bogachev2007a}, there exists a function $\bar{G}_i(P)$ with $0 \le \bar{G}_i(P) \le 1$ such that
\begin{align}
&  \int \d \mu (g)  R_V(g) u (P) G_i^{U^\dagger(g) E U(g)} [V^{\dagger} (g) P V(g)] = \nonumber \\ & \qquad \qquad \bar{G}_i (P) \int \d\mu (g)  u [V^{\dagger}(g)P V(g)].
\label{eq:mean_theorem} 
\end{align}
Let $u^{\ast}(P)= \int \d\mu (g)  u [V^{\dagger}(g)P V(g)]$, which is obviously $R_V-$covariant. Then
\begin{equation}
E_i'= \int \d \omega (P) u^{\ast} (P) \bar{G}_i(P).
\end{equation}
To see that $\bar{G}$ satisfies the normalisation, we sum~\eqref{eq:mean_theorem} over $i$:
\begin{equation}
\left[\sum_{i=1}^{n} \bar{G}_i (P) \right] u^{\ast} (P) = u^{\ast} (P), 
\end{equation}
which means $\sum_{i=1}^{n} \bar{G}_i (P) =1$ almost everywhere with measure generated by $u^{\ast}(P)$. Therefore $\bar{G}$ is a proper response function for measurement $E$ with LHS ensemble $u^{\ast}$. Thus $\rho$ is also $u^{\ast}$-unsteerable.
\end{proof}

One can easily check that both proofs work equally well when the measurements are restricted to PVMs. To better understand the relation of Theorem~\ref{th:second} with Lemma 1 of Ref.~\cite{wiseman2007steering}, we come back to the example of the two-qubit Werner state $W_{p}$ in the main text. According to Lemma 1 of Ref.~\cite{wiseman2007steering}, if $W_p$ is unsteerable, there exists a covariant indexed LHS ensemble for which $W_{p}$ is unsteerable, and then it is deduced that this singles out the uniform distribution over the hidden variables as the ``optimal'' LHS ensemble $u^{\ast}$. In fact, there exist infinitely many indexed ensembles that are covariant under the $\mathrm{U}(2)$ action. For example, consider and index space $\Lambda= Z_K \times \S_B$ with $Z_K=\{0,1,2,...,K-1\}$, i.e., we have a composite hidden variable $\lambda=(\alpha,P)$ with $\alpha=0,1,2,...,K-1$ and $P \in \S_B$. The measure $\nu$ on $\Lambda$ is generated by the distribution $u$ on $\Lambda$, defined by $u(\alpha,P)= c_\alpha$ independent of $P$ with $c_\alpha\ge 0$, $\sum_{\alpha=0}^{K-1} c_\alpha=1$.

The action of $g \in \mathrm{U}(2)$ on $\Lambda$ can be defined as $g\lambda= g(\alpha,P)= (\alpha,g P g^{\dagger})$. This action is not transitive on $\Lambda$; it has $K$ orbits indexed by $\alpha$. As a result, despite the fact that $u$ is $\G$-invariant, $u(\lambda)=u (g\lambda)$, it is not uniform over $\Lambda$ if $c_\alpha$ are distinct numbers. 

Now consider the indexed LHS ensemble given by the indexing function $F:\Lambda \to \S_B$, $(\alpha,P) \mapsto F(\alpha,P)=P$ (which is many-to-one). The LHS ensemble is apparently covariant, that is, $F(\lambda) = g^{\dagger} F(g \lambda) g$ or even $u({\lambda}) F(\lambda) = u(g\lambda) g^{\dagger} F(g \lambda) g$.

Being unsteerable with respect to this indexed LHS ensemble implies that for any $n$-POVM $E$, there exist response functions $G_i(\alpha,P)$ such that
\begin{align}
E_i' &= \sum_{\alpha=0}^{K-1} \frac{1}{4 \pi}\int \d S (P) G_i (\alpha,P) u (\alpha,P) P \nonumber \\
&= \sum_{\alpha=0}^{K-1}  \frac{1}{4 \pi} \int \d S(P) G_i (\alpha,P) c_{\alpha} P.
\end{align}
where $S$ is surface measure on Bob's Bloch sphere. On the face of it, this does not imply that one can choose the uniform distribution on Bob's Bloch sphere to be the LHS ensemble. The latter requires that there exists response function $\bar{G}_i(P)$ such that 
\begin{align}
E_i' = \frac{1}{4 \pi} \int \d S(P) \bar{G}_i (P) P.
\end{align}
The existence of $\bar{G}_i (P)$ only follows upon applying the mean value theorem as in Proof 2, which states that there exist $\bar{G}(P)$ such that 
\begin{equation}
\bar{G}(P)= \sum_{\alpha=0}^{K-1} c_\alpha G(\alpha,P).
\end{equation}
More complicated examples can be easily constructed by replacing $Z_K$ with a more complicated measurable space. Corollary~\ref{cor:1} then implies that all these different constructions are actually equivalent when one concerns with simulating steering assemblages, since they generate the same uniform distribution on the Bloch sphere. Thus one sees that Lemma 1 of Ref.~\cite{wiseman2007steering}, when augmented with Corollary~\ref{cor:1}, can also identify the uniform distribution as the optimal choice for LHS ensemble as stated directly in Theorem 2. 

\section{Simulated annealing and computation of the gap function}
\label{sec:simulated_annealing}
Simulated annealing is a standard heuristic algorithm to solve a generic global optimisation problem~\cite{press1989numerical}. In our case, we wish to compute
\begin{equation}
\Delta = \min_{Z \in \C^4,E \in \N^4} F(Z,E)
\end{equation}
with 
\begin{equation}
F(Z,E)= \frac{1}{4 \pi}\int \d S (P) \max_{i} \dprod{Z_i}{P} - \sum_{i=1}^{4} \Tr [\rho (Z_i \otimes E_i)],
\label{eq:objective_func}
\end{equation}
where $S$ is the sureface measure of the Bloch sphere (which is different from the Haar measure by a factor $\frac{1}{4 \pi}$).
The simulated annealing algorithm goes as follows. One first regards $F(Z,E)$ as an energy function of a system in the state space $(Z,E)$. 
Simulated annealing couples this system to an effective heat bath, whose temperature is then lowered slowly, so that configurations with decreasing energy are explored. 
At each temperature the system follows stochastic dynamics leading to equilibrium with the heat bath. 
The system is cooled down slowly to sufficiently small temperature $T_f$. 
It is known that if the temperature schedule is sufficiently slow then the system converges to a global minimum of $F(Z,E)$~\cite{Granville1994a}.
However, the required cooling schedule is too slow that it is not useful in practice and an alternative cooling schedule is used. 
Here we use an exponential cooling scheme, i.e., in each step the temperature is cooled down by a factor $f$.
The system can in principle become stuck in a local minimum at $T_f$. It is then necessary to repeat the cooling procedure multiple times. 

\begin{figure}[t]
	\includegraphics[width=0.45\textwidth]{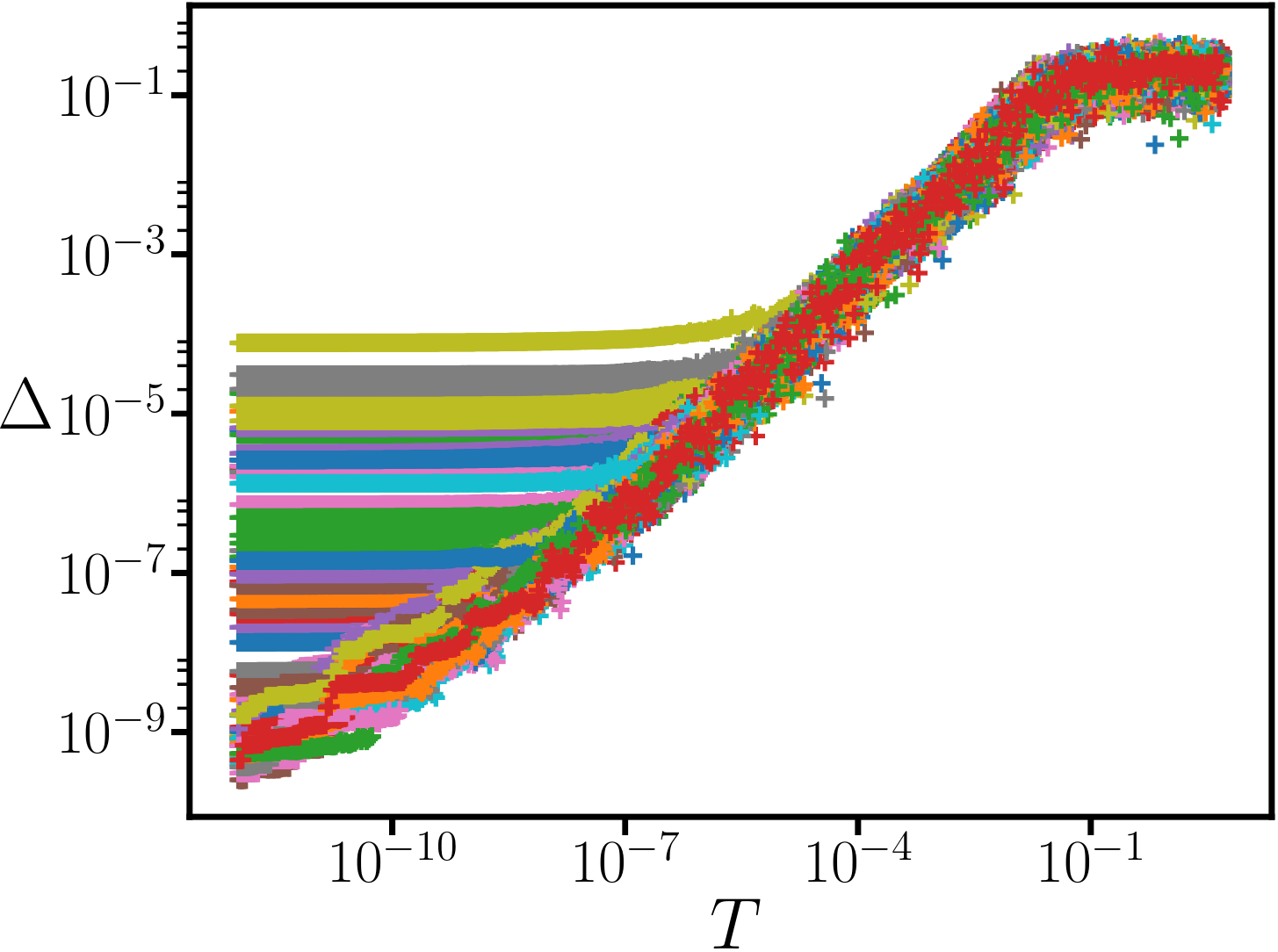}
	\caption{Typical trajectories of the energy of the system during the cooling procedure. The trajectories of $64$ different replicas are plotted in different colours. The data is for a Werner state with $p=0.49$.}
	\label{fig:werner_series}
\end{figure}

\textit{Coordinisation of variables.} Note that for two qubit systems, $\A=\B=M(\CC,2)$, where  $M(\CC,2)$ is the algebra of $2 \times 2$ complex matrices. We use the Pauli basis $\{\sigma_i\}_{i=0}^{3}=\{\II,\sigma_x,\sigma_y,\sigma_z\}$ to coordinate the real subspace $M^{H}(\CC,2)$. Each operator $X \in M^H(\CC,2)$ is therefore characterised by $4$ (real) coordinates $x_i$,
\begin{equation}
X= \frac{1}{2} \sum_{i=0}^{3} x_i \sigma_i.
\end{equation}  
The boundary of the positive cone of $M^{H}(\CC,2)$ is given by $x_0^2- x_1^2-x_2^2-x_3^2=0$ with $x_0 \ge 0$, consisting of vectors of the form $\alpha \begin{pmatrix} 1 \\ \n \end{pmatrix}$ with $\alpha \ge 0$.

The composite operator $Z=\oplus_{i=1}^{4} Z_i$ and $E=\oplus_{i=1}^{4} E_i$ are thought of as $4 \times 4$ matrices, in which each column is the coordination of $Z_i$ and $E_i$ respectively. From now on, we will use $Z$ and $E$ to denote these matrices, and $Z_i$, $E_i$ to denote the $i^\text{th}$ column. 

To implement the constraint $\C^4$ to $Z$, we write \mbox{$Z=XR$}, where
\begin{equation}
R=\frac12\begin{pmatrix}
1 & -1 & -1 & 1 \\
-1 & -1 &  1 & 1 \\
-1 & 1 & -1 & 1 \\
1 & 1 & 1 &  1
\end{pmatrix},
\end{equation}  
and $X$ satisfies $\Tr (X^{T} X)=2$, $X_4=(0,0,0,0)^T$.

\begin{figure}[t!]
	\begin{center}
		\includegraphics[width=0.48\textwidth]{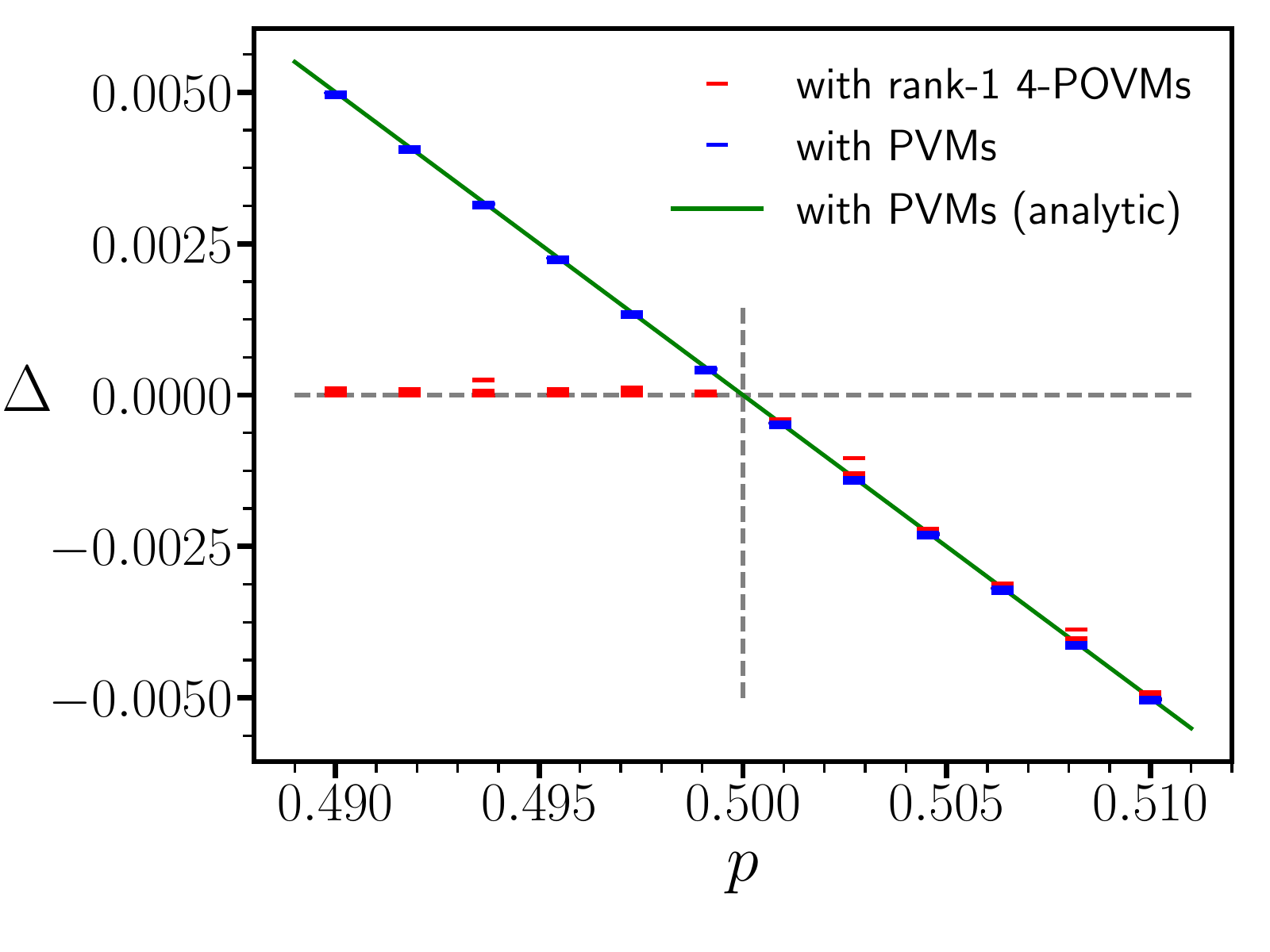}
	\end{center}
	\caption{The gap function for the Werner state with PVMs and rank-1 4-POVMs from $512$ replicas. Note that $512$ trials strongly concentrate at the minimum values, which gives confidence to the algorithm. 
		}
	\label{fig:werner_replicas}
\end{figure}

To implement the constraint $\N^4$ on $E$, we note that every component $E_i$ is on the boundary of the positive cone, 
$E_i = \alpha_i \begin{pmatrix} 1 \\ \n_i \end{pmatrix}$, where $\alpha_i \ge 0$ and $\pmb{n}_i$ is a unit vector.
The constraint $\sum_{i=1}^{4} E_i =\II$ can be considered as a constraint for $\{\alpha_i\}_{i=1}^{4}$. In fact, if $\begin{pmatrix} 1 \\ \n_i \end{pmatrix}$ are independent, $\{\alpha_i\}_{i=1}^{4}$ are uniquely determined. Note that the set of $E$ where $\begin{pmatrix} 1 \\ \n_i \end{pmatrix}$ are dependent are zero-measured in $\N^4$. In practice, we therefore do not need to worry about the case that $\begin{pmatrix} 1 \\ \n_i \end{pmatrix}$ are dependent if the linear solver is relatively stable. Here we use the Householder linear solver, provided by \texttt{Eigen 3}~\cite{eigenweb}. Further, we use a common technique~\cite{press1989numerical} to take care of the constraint $0 \le \alpha_i \le 1$ by assigning infinite values to the energy if the solution $\alpha_i$ are outside $[0,1]$. This allows us to describe the (extended) set $\N^4$ by $4$ vectors $\{\n_i\}_{i=1}^{4}$.

\textit{Annealing.} At temperature $T$, the stochastic dynamics of the system is simulated by the Metropolis algorithm~\cite{press1989numerical}: at every time step, the system tries an elementary step, which will be accepted with probability $\min\{1, e^{-\Delta F/T}\}$, where $\Delta F$ is the change in energy due to the trial step. In each elementary trial step, either $Z$ or $E$ is updated with equal probability. If $Z$ is updated, we choose randomly two elements of $X$, say $X_{ij}, X_{kl}$, where $j,l <4$, and perform a rotation $Q(\theta) \in SO$(2) on the vector $(X_{ij}, X_{kl})^T$ by a random angle $\theta$ normally distributed with mean $0$ and standard deviation $2\pi\sqrt{T}$. The components of the vector $Q(\theta)(X_{ij}, X_{kl})^T$ replace the $ij$ and $kl$ elements of $X$. Note that the constraints $\Tr (X^{T} X)=2$ and $X_4=(0,0,0,0)^T$ are respected in the new $X$. Then $Z$ is updated as $Z=XR$. If $E$ is updated, we choose one of the vectors $\n_i$ randomly, and rotate it around one of the $3$ axes $x$, $y$, $z$ by a random angle normally distributed with mean $0$ and standard deviation $2\pi\sqrt{T}$.  At each temperature, the number of the simulated steps are at least $100$ times the degree of freedom. 

\begin{figure}[t!]
	\begin{center}
		\includegraphics[width=0.45\textwidth]{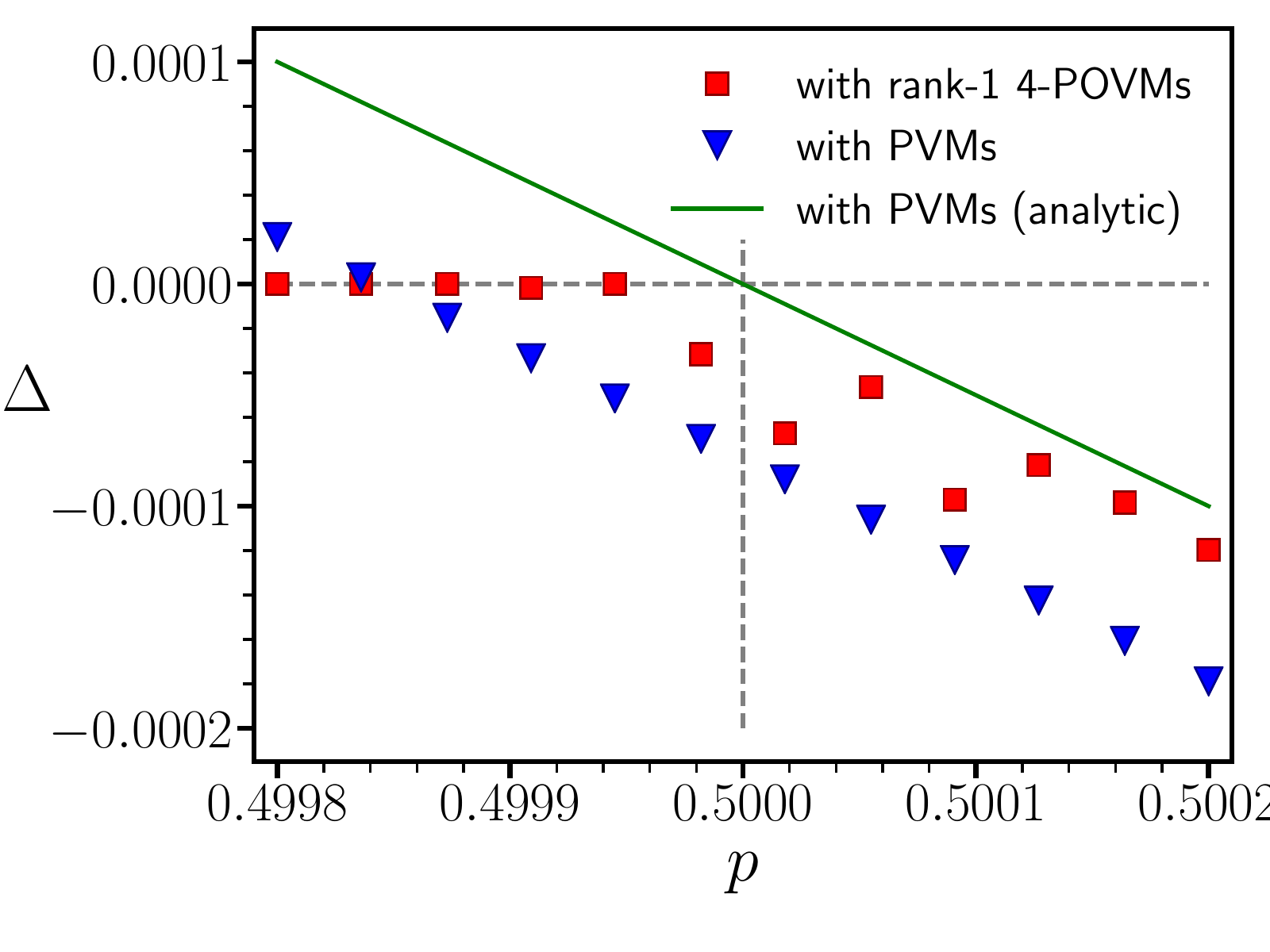}
	\end{center}
	\caption{The gap function for the Werner state with PVMs and ePOVMs at high resolution around the transition point $p=\frac{1}{2}$.}
	\label{fig:werner_fine}
\end{figure}

\textit{Cooling schedule.} After annealing the system at temperature $T$, the temperature is decreased by a factor $f$; here $f=0.95$. This is known as exponential temperature scheduling~\cite{press1989numerical}. The initial temperature $T_i$ is chosen as the maximal value minus the minimal value of the energy function sampled at $1000$ times the degree of freedom points. The algorithm is stopped at temperature $T_f=10^{-9}$. Lowering the final temperature does not significantly improve the results.

\textit{Replicas.} We repeat the cooling procedure $M=512$ times. 64 such typical cooling trajectories are presented in Figure~\ref{fig:werner_series}. As shown in Figure~\ref{fig:werner_replicas}, all $512$ replicas produce a very similar minimum energy, suggesting that there is no major local minimum in the energy landscape. This provides confidence that the system indeed converges close to a global minimum.

\textit{Numerical accuracy.} As seen in Figure~\ref{fig:werner_series}, the gap can be \emph{overestimated} by some order of $10^{-4}$. This makes it difficult to study very small gap values when $p$ is around $\frac{1}{2} \pm 10^{-4}$. Other cooling schedules~\cite{press1989numerical}, parallel tempering~\cite{swendsen1986replica} or more subtle global optimisation techniques~\cite{Arora1995a} can be considered to increase the accuracy. However, the numerical accuracy is limited by another critical factor: the accuracy of the spherical integral in~\eqref{eq:objective_func}. Here, we used Lebedev's quadrature with $5810$ points to compute spherical integrals. This effectively replaces the optimal uniform LHS ensemble by a suboptimal discrete distribution at $5810$ quadrature points. Accordingly, the expected transition probability $p$ is shifted by some value of $2 \times 10^{-4}$ to the left of $\frac{1}{2}$, smearing out the accuracy of the simulated annealing optimisation as seen in Figure~\ref{fig:werner_fine}. Similar problems occur for $T$-states at a resolution of $\abs{\epsilon} \approx 10^{-3}$ around the surface of unsteerable states.

\section{The gap function for steering the Werner states with PVMs}
\label{sec:PVM_gap}

The numerical calculation of the gap function for the Werner states with PVMs is carried out similarly to Appendix~\ref{sec:simulated_annealing}. The analytical calculation of the gap function for steering a Werner state with PVMs is rather straightforward. We are to calculate
\begin{multline}
\Delta = \min_{(Z_1,Z_2) \in \C^2} \left\{ \frac{1}{4 \pi} \int \d S (P) \max \{ \dprod{Z_1}{P}, \dprod{Z_2}{P}\} - \right. \\ \left. \max_{(P_1,P_2)} \Tr [W_p (P_1 \otimes Z_1 + P_2 \otimes Z_2)] \right\},
\end{multline}
where $(P_1,P_2)$ forms a projective measurement, i.e., $P_1$, $P_2$ are orthogonal projections such that $P_1+P_2=\II$.

Since $(Z_1,Z_2) \in \C^2$ implies that $Z_1+Z_2=0$, we can set $Z_1=X$ and $Z_2=-X$. Moreover, because of the $\operatorname{U}(2)$ symmetry of the problem, we can suppose $X= \lambda_0 \ketbra{0}{0} + \lambda_1 \ketbra{1}{1}$ with $\lambda_0 \ge \lambda_1$. Because $\dprod{Z_1}{Z_1} + \dprod{Z_2}{Z_2}= 1$, we have $\dprod{X}{X}= \frac{1}{2}$, or $\lambda_0^2+\lambda_1^2=\frac{1}{2}$.

If we write the projections in Pauli coordinates as $P=\begin{pmatrix} 1 \\ \n \end{pmatrix}$, then $\dprod{Z_1}{P} \ge \dprod{Z_2}{P}$ is equivalent $n_z \ge -\frac{a}{b}$, with $a=\lambda_0+\lambda_1$ and $b=\lambda_0-\lambda_1$. Therefore
\begin{multline}
\frac{1}{4 \pi} \int \d S (P) \max \{ \dprod{Z_1}{P}, \dprod{Z_2}{P}\} = \\  \frac{1}{4 \pi} \int_{n_z \ge -\frac{a}{b}} \d S (P)  \dprod{X}{P} - \int_{n_z \le -\frac{a}{b}} \d S (P)  \dprod{X}{P},\\
\end{multline}
which evaluates to $\frac{1}{4b}$. 

On the other hand 
\begin{multline}
\max_{(P_1,P_2)} \Tr [W_p (P_1 \otimes Z_1 + P_2 \otimes Z_2)] = \\
\max_{P_1} \Tr [W_p (P_1 \otimes X)] - \Tr[W_p (\II \otimes X)],
\end{multline}
which evaluates to $\frac{pb}{2}$.

Thus we have $\Delta = \min_{b} \left\{ \frac{1}{4b} - \frac{pb}{2}\right\}$. Note that $\lambda_0^2+\lambda_1^2=\frac{1}{2}$ implies that $a^2+b^2=1$, and so $b\le 1$. Hence we obtain the gap function for PVMs as $\Delta = \frac{1}{4}-\frac{p}{2}$.
\bibliography{steering-geometry}

\end{document}